\documentclass[10pt]{article}
\usepackage{xspace}
\usepackage{xcolor}
\usepackage{amsmath,amssymb,enumitem}
\usepackage{hyperref}

\newcommand{\NULL}{\text{{\tt NULL}}\xspace}

\newcommand{\ftab}{\text{\sc D}}
\newcommand{\gtab}{\text{\sc G}}

\newcommand{\PV}{\text{$\mathcal{P}\mathcal{V}$}\xspace}

\newcommand{\realizza}{\rhd} 
\newcommand{\nonrealizza}{\ntriangleright} 

\newcommand{\forza}{\Vdash} 
\newcommand{\nonforza}{\nVdash}
\newcommand{\kripke}{\mbox{$\langle P,$}\mbox{$\leq,$}\mbox{$\rho$,}\mbox{$\forza\rangle$}}

\newcommand{\true}{\mbox{$\top$}\xspace}
\newcommand{\false}{\mbox{$\bot$}\xspace} 

\newcommand{\soddisfa}{\models}

\newcommand{\T}{\text{$\bf T$}}

\newcommand{\F}{\text{$\bf F$}}
\newcommand{\Ttilde}{\text{$\bf \widetilde{\T}$}}
\newcommand{\That}{\text{$\bf \widehat{\T}$}}

\newcommand{\Fu}{\text{$\bf F_l$}\xspace}
\newcommand{\Fn}{\text{$\bf F_n$}\xspace}
\newcommand{\Tn}{\text{$\bf T_n$}\xspace}
\newcommand{\To}{\text{$\bf \overline{\T}$}}

\newcommand{\dum}{\mbox{$\mathbf{Dum}$}\xspace}

\newcommand{\uK}{\mbox{$\underline{K}$}\xspace}

\newcommand{\Rcal}{\mathcal{R}}

\newcommand{\calS}{\mathcal{S}}
\newcommand{\calR}{\mbox{$\mathcal{R}$}\xspace}

\newcommand{\EE}{|}

\newcommand{\npwt}{\text{$\mathbf{D_1}$}\xspace}
\newcommand{\set}{\text{$\mathbf{D_2}$}\xspace}
\newcommand{\et}{\text{$\mathbf{D_3}$}\xspace}

\def\Tab#1#2#3{
       \frac{\phantom{a}\stackrel{\textstyle #1}
       {\phantom{\scriptscriptstyle .}}\phantom{a}}
       {\stackrel{\phantom{\scriptscriptstyle .}}{\textstyle #2}}
{\scriptstyle #3} 
 }

\def\NewTab#1#2#3{\frac{ #1}{ #2}{\scriptstyle #3}}
\title{
Terminating Calculi for Propositional Dummett Logic with Subformula Property
}
\author{
 Guido Fiorino\\
 Dipartimento di Metodi Quantitativi per le Scienze Economiche ed Aziendali,\\
 Universit\`a di Milano-Bicocca, Piazza dell'Ateneo Nuovo, 1, 20126
 Milano, Italy.\\
 guido.fiorino@unimib.it
}

\newtheorem{proposition}{Proposizione}
\newtheorem{lemma}{Lemma}

\newtheorem{theorem}{Theorem}
\newtheorem{remark}{Remark}
\newtheorem{proof}{Proof}
\newtheorem{claim}{Claim}

\begin{document}
\maketitle
\begin{abstract}
In this paper we present  two terminating tableau calculi for
propositional Dummett logic obeying the subformula property. 
The ideas of our calculi rely on the linearly ordered Kripke semantics
of Dummett logic. 
The first calculus  works on two semantical levels:
the present and the next possible world. 
The second calculus employs the usual object language of tableau systems and exploits 
a property of the construction of the completeness theorem to introduce a check which is
an alternative to loop check mechanisms.
\end{abstract}


\pagestyle{plain}

\section{Introduction}
In this paper we present  two terminating tableau calculi for
propositional Dummett logic obeying the subformula property. The depth of the deductions
of the first calculus  is quadratic and allow to extract a counter model whose depth is
$n+1$ at most, with $n$ the number of propositional variables in the formula to be
decided. The depth of the deductions of the second calculus is linear. 
To avoid the introduction of loop check mechanisms, 
our calculi exploit   the linearly ordered Kripke semantics
of Dummett logic. 
The first calculus uses the ideas presented in paper~\cite{Fiorino:2011}, 
and works on two semantical levels:
the present and the next possible world. 
The second calculus uses the usual $\T$ and $\F$ signs and exploits 
a property of the construction of the completeness theorem to introduce a check which is
an alternative to loop check mechanisms.

Dummett logic has been extensively investigated both by people working in
computer science and in logic.
 The history of this  logic starts with G\"odel, who studied the family of
 logics semantically characterised by a sequence of $n$-valued ($n>2$)
 matrices (\cite{Goedel:86}).  
 In paper~\cite{Dummett:59} Dummett studied the logic semantically characterized
 by an infinite valued matrix which is included in the family of logics studied by
 G\"odel and proved that such a logic is axiomatizable by adding to any Hilbert
 system for propositional intuitionistic logic the axiom scheme $(p\to q)\lor
 (q\to p)$. Moreover, it is well-known that such a logic is semantically
 characterised by linearly ordered Kripke models. Dummett logic also appear in
 investigations related to 
 the relevance logics~\cite{DunMey:71} and Heyting provability~\cite{Visser:1982}. 
 Dummett logic   has been studied also in recent years for its
 relationships with computer science~(\cite{Avron:91}) and fuzzy
 logics~(\cite{Hajek:98}). For a survey of proof theory for G\"odel-Dummett logics
 we cite~\cite{BaaCiaFer:2003}.

To perform automated deduction  
both tableau and sequent calculi have been proposed. 
To get a terminating calculus for Dummett logic obeying to the subformula
property, the main problem is  how to handle formulas of the kind $\T (A\to B)$ 
(left-implicative formulas, in the sequent terminology). A terminating calculus can be achieved 
by introducing specialized rules based on the main connective of $A$. 
In the conclusions of such specialized rules some formulas are not subformulas of the 
premise. 
Calculi of this nature are provided in 
\cite{AveFerMig:99,AvrKon:2001,Dyckhoff:99,Fiorino:00a,Fiorino:2010,Fiorino:2011,Larchey-Wendling:2007}. 
The specialized rules used in~\cite{AveFerMig:99,Dyckhoff:99,Fiorino:00a,Fiorino:2010,Fiorino:2011} 
are based on the
rules proposed by
Vorobiev~\cite{Vorobiev:58} to handle formulas of the kind  $\T(A\to B)$ in propositional
Intuitionistic logic. Papers~\cite{AvrKon:2001,Larchey-Wendling:2007} decompose implicative
formulas by rules whose correctness is justified by the semantics of Dummett logic.

In this work we present calculi whose deductions have,  respectively, linear and quadratic
depth in the size of the formula to be decided and the subformula
property, a feature that the calculi in the above quoted papers fail.

Papers~\cite{BaaFer:99} and~\cite{MetOliGab:2003} provide calculi with the subformula
property. Work~\cite{BaaFer:99} provides a calculus based on sequents called 
{\em sequent of relations
calculus} whose deductions have
exponential depth in the formula to be proved, 
because in the premise of some rules can occur multiple copies of a subformula of the 
conclusion. Moreover, the nodes of a proof with such a system are 
more cumbersome than the nodes of a tableau proof, because every node of the deduction
expresses the  relation order between the subformulas of the formula to be proved. Thus
every node has a quadratic number of formula occurrences. 
Paper~\cite{MetOliGab:2003} provides two goal-oriented calculi, one  based on
hypersequents and one on labelled sequents. The systems are restricted to the implicative
fragment. The first advantage is that our results are given for the
full language. Although a  translation from the full language to the implicative fragment is
possible but, in the case of disjunction  the cost is an exponential blow-up in the
size of the formula. Moreover,  we do not need the more expressive power of hypersequents and
differently from the labelled sequents,  
the object language of our calculi does not depend on the input.

As regards our results, it is worth to remark that our final calculus is a genuine tableau
calculus only employing the usual two signs $\T$ and $\F$ corresponding respectively to the
left-hand side and right-hand side of the sequent systems. 

As regards the techniques used in the paper, we have deliberately chosen
to employ tableaux as proof-systems. Our choice is justified by the fact that the rules
of our calculi can be easily explained by semantical 
considerations based on Kripke models. For this reason we do not use the
sequent systems, whose behaviour is upside-down with respect to a semantical
characterization. However, since tableau and sequent calculi are related, it is an easy
exercise to translate our calculi into sequents. Moreover, correctness and completeness are proved 
always taking the Kripke models for propositional Dummett logic as semantical reference.   
Following the proofs of the completeness theorems, the  procedures we provide can
be modified  to return a proof or a counter model. 
   
\section{Basic definitions and a terminating tableau calculus with the subformula property}
\label{sec:basic-definitions}
We consider the propositional language based on a denumerable set of
propositional variables \PV, the boolean constants \true and
\false and the logical connectives $\land,\lor,\to$.
We call {\em atoms} the elements of $\PV\cup \{ \true,\false\}$. 
In
the following, formulas (respectively  set of formulas and  
propositional variables) are denoted by 
letters $A$, $B$, $C$\dots (respectively $S$, $T$, $U$,\dots and $p$, $q$,
$r$,\dots ) possibly with subscripts or superscripts.  

From the introduction we recall that  {\em Dummett Logic}
($\dum$) can be axiomatized by adding to any axiom system for
propositional intuitionistic logic the axiom scheme $(p\to q)\lor
(q\to p)$ and a  well-known semantical
characterization of $\dum$ is by 
{\em linearly ordered Kripke models}. 
In the paper {\em model}
means a linearly ordered Kripke model, namely a structure
$\underline{K}=\kripke$, where $\langle P,\leq,\rho 
\rangle$ is a linearly ordered set with $\rho$ minimum with respect to
$\leq$ and $\forza$ is the {\em forcing
  relation}, a binary relation on $P\times(\PV\cup \{\true,\false \})$
such that: (i) if $\alpha\forza p$ and $\alpha \leq \beta$, then  $\beta\forza p$;
(ii) for every $\alpha\in P$, $\alpha\forza \true$ holds and
$\alpha\forza \false$ does not hold. Hereafter we denote the members of $P$ by
means of lowercase
letters of the Greek alphabet.

The forcing relation is extended in a standard way to arbitrary formulas of
our language  as follows:
\begin{enumerate}
\item $\alpha\forza A\land B$ iff $\alpha\forza A$ and $\alpha\forza B$;
\item $\alpha\forza A\lor B$ iff $\alpha\forza A$ or $\alpha\forza B$;
\item $\alpha\forza A\to B$ iff, for every $\beta\in P$ such that $\alpha\leq\beta$,
  $\beta\forza A$ implies $\beta\forza B$;
\end{enumerate}
We write $\alpha\nonforza A$ when $\alpha\forza A$ does not hold.  
It is easy to prove that for every formula $A$ the {\em persistence}
property holds: If $\alpha\forza A$
and $\alpha\leq \beta$, then $\beta\forza A$. We say that $\beta$ is {\em immediate
successor} of $\alpha$ iff $\alpha<\beta$ and there is no $\gamma\in P$ such that
$\alpha<\gamma<\beta$.
A formula \emph{$A$ is valid in a model $\uK=\kripke$} if and only if $\rho
\forza A$.  It is well-known that
\dum coincides with the set of formulas valid in all models.

In Figures~\ref{fig:invertible-rules} and~\ref{fig:non-invertible-rule} are given the
rules of \npwt,  a terminating tableau calculus exploiting the truth at present and next possible world
in the Kripke semantics.  

\begin{figure}[ht]
{
  \[
  \begin{array}[t]{|c|}
    \hline
    \begin{array}{ ccc }
      \Tab{S, \T (A \land B) }{S, \T A ,  \T B }{\T\land} 
      &
      &
      \Tab{S,\Fu (A\land B)}{S,\Fu A,\Tn B|S,\T A,\Fu B}{\Fu\land}
      
    \end{array}
    \\\\
    \begin{array}[t]{ccc}
      \Tab{S, \T (A \lor B) }{S, \T A | S, \T B }{ \T\lor} 
      &
      &
      \Tab{S,\Fu(A\lor B)}{S,\F A,\Fu B|S,\F B, \Fu A}{\Fu\lor} 
    \end{array}
    \\\\
    \begin{array}{cc}
      \Tab{S, \T (A \to B) }{S, \T B | S, \Fu A,\Tn B | S,\Ttilde (A\to B) }{
        \T\to} 
      &
      \Tab{S,\Fu(A\to B)}{S,\T A,\Fu B}{\Fu\to}
    \end{array}
    \\\\
    \begin{array}{cc}
      \Tab{S,\F A}{S,\Fu A | S,\Fn A}{\F-decide}
      &
      \Tab{S,\That(A\to B)}{S,\Fu A,\Tn B | S,\Ttilde (A\to B)}{\That-decide}
    \end{array}
    \\\\
    \hline
  \end{array}
  \]  
}
  \caption{The invertible rules of \npwt.}
  \label{fig:invertible-rules}
\end{figure}
\begin{figure}[ht]
  {
    \[
    \begin{array}{| c |}
      \hline
      \Tab
      {
        S, \Ttilde(A_1\to B_1),\dots,\Ttilde(A_n\to B_n),\Fn C_{n+1} ,\dots, \Fn C_u 
      }
      { 
        S_c,V_1,S_{\F}|\dots|S_c,V_n,S_{\F}|S_c,S_{\That},V_{n+1}|\dots|S_c,S_{\That},V_u
      }
      { 
        \Fn\Ttilde
      }
      \\\\
      \text{where:}\  \ 
      S_{\F}=\{\F C_{n+1} ,\dots, \F C_u \},\ \ 
      S_{\That}=\{\That(A_1\to B_1),\dots,\That(A_n\to B_n)\},\hspace*{\fill}\\[1ex]
      \text{for }j=1,\dots,n,\hspace*{\fill}\\[1ex] 
      V_j = \{ \That(A_1\to B_1),\dots,\That(A_{j-1}\to B_{j-1}),
      \Fu A_j,\Tn B_j,
      \That(A_{j+1}\to B_{j+1}),\dots,\That(A_n\to B_n)\},
      \hspace*{\fill}\\[1ex]
      \text{for }j=n+1,\dots,u,\ 
      V_j = \{\F C_{n+1},\dots,\F C_{j-1},\Fu C_j, \F
      C_{j+1},\dots,\F C_u \} \text{ and }\hspace*{\fill}
      \\[1ex]
      S_c = \{ \T A | \T A \in S \} \cup \{\T A | \Fu A\in S\}\cup
      \{\T A | \Tn A\in S\};\hspace*{\fill}       
      \\[1ex]
      \hline
    \end{array}
    \]
  }   
  \caption{The non-invertible rule $\Fn\Ttilde$.}
  \label{fig:non-invertible-rule}
\end{figure}
The calculus \npwt works on signed formulas, that is
well-formed formulas prefixed with one of the {\em signs}
$\T$ (with $\T A$ to be read ``the fact $A$ is known at the present state of
knowledge''), $\F$ (with $\F A$ to be read ``the fact $A$ is not known at the
present state of knowledge''), $\Fu$ (with $\Fu A$ to be read ``this is the last
state of knowledge where $A$ is not known''), $\Fn$ (with $\Fn A$ to be read ``$A$
is not known in the next state of knowledge''), $\Tn$ (with $\Tn A$ to be
read ``$A$ will be known in the next state of knowledge''), 
$\That$ (with $\That A$ to be read as ``\T A holds and if $A$ is
of the kind $B\to C$, then $\F B$ holds'') and $\Ttilde$ (with $\Ttilde
A$ to be read as ``\T A holds and if $A$ is 
of the kind $B\to C$, then $\Fn B$ holds'')  
and on
sets of signed formulas (hereafter we 
omit the word ``signed'' in front of ``formula'' in all the contexts
where no confusion  arises). 
Formally, the meaning of the signs is provided by the relation 
{\em  realizability}  ($\realizza$)
defined as follows:
Let $\underline{K}=\kripke$ be a model, let $\alpha \in P$,
let $H$ be a signed formula and let $S$ be a set of signed
formulas. We say that $\alpha$ {\em realizes} 
$H$ (respectively $\alpha$ {\em realizes } $S$ and $\uK$ {\em realizes} $S$), and we write $\alpha
\realizza H$ (respectively $\alpha\realizza S$ and $\uK\realizza S$), if the
following conditions hold:
\begin{enumerate}
\item $\alpha\realizza \T A$ iff $\alpha\forza A$;
\item $\alpha\realizza \F A$  iff $\alpha\nonforza A$;
\item $\alpha\realizza \Fn A$  iff there exists $\beta>\alpha$, $\beta\realizza
  \F A$;
\item $\alpha\realizza \Tn A$  iff for every $\beta>\alpha$, $\beta\realizza \T A$;
\item $\alpha\realizza \Fu A$  iff $\alpha\realizza \F A$ and $\alpha\realizza
  \Tn A$;
\item $\alpha\realizza \Ttilde A$  iff $A\equiv B\to C$ and $\alpha\realizza \T A$ and $\alpha\realizza \Fn B$;
\item $\alpha\realizza \That A$  iff $A\equiv B\to C$ and $\alpha\realizza \T
  A$ and $\alpha\realizza \F B$;

\item $\alpha\realizza S$  iff $\alpha$ realizes every formula in $S$.
\end{enumerate} 

By inspecting the rules of the calculus we have that signs  $\Ttilde$ and $\That$ are
used for  implicative formulas only.

From the meaning of the signs we get the conditions that make a set of
formulas inconsistent. A set $S$ is {\em inconsistent} if one of the
following conditions 
holds:\\ 
-$ \T \false\in  S$; \\
-$ \{\T A, \F A\} \subseteq  S$;\\
-$ \{\T A, \Fu A\} \subseteq  S$.
 
It is easy to prove the following proposition:
\begin{proposition}
  \label{prop:contraddittorieta}
  If a set of formulas $S$ is  inconsistent, then for every Kripke model
  $\uK=\kripke$ and for every $\alpha\in P$, $\alpha\nonrealizza S$.
\end{proposition}

  

A proof table (or proof tree)
for $S$ is a tree, rooted in $S$ and obtained by the subsequent 
instantiation of the rules of the calculus. 
A {\em closed proof table} is a proof table whose leaves are all inconsistent sets. A
closed proof table is a proof of the calculus and a formula $A$ is provable iff
there exists a closed proof table for $\{\Fu A \}$.

The premise of the rules are instantiated  in a
duplication-free style: in the application of the rules we always
consider that the formulas in evidence in the premise are not in $S$.
We say that a rule \calR
applies to a set $U$ when it is possible to instantiate the premise of
\calR with the set $U$ and we say that a rule \calR applies to a
formula $H\in U$ (respectively the set $\{H_1,\dots,H_n\}\subseteq U$)
to mean that it is possible to instantiate the premise 
of \calR taking $S$ as $U\setminus \{H \}$ (respectively $U\setminus
\{H_1,\dots,H_n\}$). 
As an example, given the set $U=\{\T(B\land C), \T(A\land C),\F(A\lor
B)\}$, by {\em  applying  the rule} $\T\land$ {\em taking } $\T(A\land C)$ {\em as main
  formula} means to instantiate the premise of $\T\land$ taking $S=\{\T(B\land
C),\F(A\lor B)\}$  and $H=\T(A\land C)$. 


Before going into technical details we give an informal description of the whole machinery. 
First, note that there are no rules for sign $\Tn$. The sign $\Tn$ aims to mark
formulas that will be signed with $\T$ after an application of $\Fn\Ttilde$.
Moreover, $\F$-formulas are handled by $\F$-decide. In semantical terms of counter model
construction, given the formula $\F A$,  rule $\F$-decide
decides if in the next state of knowledge the formula $A$ will be  a known or an unknown fact.
Similarly, rule $\That$-decide decides the semantical status of the antecedent $A$ for formulas of the
kind $\That(A\to B)$. By the rules of the calculus, if $\T(A\to B)$ becomes  $\Ttilde(A\to
B)$, then in the subsequent sets the sign of $A\to B$ can only be  $\That$ or $\Ttilde$
and rules $\That$-decide and $\Fn\Ttilde$ are the only rules where the sign can be switched.  
Rule $\Fn\Ttilde$ is the only non-invertible rule of \npwt, thus to devise a
complete strategy that does not require backtracking it is sufficient that rule
$\Fn\Ttilde$ is applied when no other rule is applicable. Finally,
the following features of \npwt allow us to prove the termination:  every node of the proof
table contains at least a $\Fu$-formula and 
 an application of $\Fn\Ttilde$ increases the number of $\T$-signed
formulas. Also note that, if $\Fn\Ttilde$ is applied only if no other rule is applicable, which
is the way we want to use $\Fn\Ttilde$, then  the
premise is always instantiated to a set containing at least an $\Fu$-atomic formula.
This implies that  in the conclusion  at least one  new $\T$-signed atomic formula is
introduced.  
Summarizing, in spite of the rightmost set in the conclusion of  rules $\T\to$, $\F-decide$ and
$\That-decide$, any sequence of application of rules ends in a set containing signed
atomic formulas only and we do not have infinite loops.

\begin{remark}
The presentation of the calculus is  without efficiency in
mind. We could exploit the meaning of signs $\Fu$ and $\Tn$ to introduce more rules and
checks that allow us to reduce the size of the proofs. As an example we could extend the
notion of inconsistent set by adding to those given above the following conditions:  
$ \{\F A, \That A\} \subseteq  S$;
$ \{\F A, \Ttilde A\} \subseteq  S$;
$ \{\T A, \Fn A\} \subseteq  S$;
$ \{\T A, \Ttilde(A\to B)\} \subseteq  S$;
$ \{\T A, \That (A\to B)\} \subseteq  S$;
$ \{\Tn A, \Ttilde (A\to B)\} \subseteq  S$;
$ \{\Tn \false, \Fn A\}\subseteq  S$; 
$ \{\Tn A, \Fn A\} \subseteq  S$;
$ \{\Fu A, \Ttilde (A\to B)\} \subseteq  S$;
$ \{\Fu \false, \Fn A\} \subseteq  S$.
This would avoid to perform useless deduction steps all ending in inconsistent sets.
The rule $\Fn\Ttilde-opt$ given 
in Figure~\ref{fig:non-invertible-rule:optimized} is an optimization of rule $\Fn\Ttilde$
of Figure~\ref{fig:non-invertible-rule}.
\begin{figure}[t]
  {
    \[
    \begin{array}{| c |}
      \hline
      \Tab
      {
        S, \Ttilde(A_1\to B_1),\dots,\Ttilde(A_n\to B_n),\Fn C_{n+1} ,\dots, \Fn C_u 
      }
      { 
        S_c,V_1,S_{\F}|\dots|S_c,V_n,S_{\F}|S_c,S_{\Ttilde},V_{n+1}|\dots|S_c,S_{\Ttilde},V_u
      }
      { 
        \Fn\Ttilde-opt
      }
      \\\\
      \text{where:}\  \ 
      S_{\F}=\{\F C_{n+1} ,\dots, \F C_u \},\ \ 
      S_{\Ttilde}=\{\Ttilde(A_1\to B_1),\dots,\Ttilde(A_n\to B_n)\},\hspace*{\fill}\\[1ex]
      \text{for }j=1,\dots,n,\hspace*{\fill}\\[1ex] 
      V_j = \{ \Ttilde(A_1\to B_1),\dots,\Ttilde(A_{j-1}\to B_{j-1}),
      \Fu A_j,\Tn B_j,
      \That(A_{j+1}\to B_{j+1}),\dots,\That(A_n\to B_n)\},
      \hspace*{\fill}\\[1ex]
      \text{for }j=n+1,\dots,u,\ 
      V_j = \{\Fn C_{n+1},\dots,\Fn C_{j-1},\Fu C_j, \F
      C_{j+1},\dots,\F C_u \} \text{ and }\hspace*{\fill}
      \\[1ex]
      S_c = \{ \T A | \T A \in S \} \cup \{\T A | \Fu A\in S\}\cup
      \{\T A | \Tn A\in S\};\hspace*{\fill}       
      \\[1ex]
      \hline
    \end{array}
    \]
  }   
  \caption{The non-invertible rule $\Fn\Ttilde$.}
  \label{fig:non-invertible-rule:optimized}
\end{figure}  
This rule avoids useless applications of $\F$-decide and $\That$-decide which are rules
that introduce branching points.
\end{remark}

\section{Correctness}
To obtain the correctness of \npwt with respect to Dummett logic, we
proceed by showing that the existence of a proof table
for $\{\Fu A\}$, implies the validity of $A$ in Dummett logic. 
The main step consists in establishing  that the
rules of the calculus preserve  
realizability:
\begin{proposition}
  For every rule of \npwt, if a world $\alpha$ of a model $\uK=\kripke$ realizes the
  premise, then $\alpha$ realizes at least one of the conclusions.  
\end{proposition}
\begin{proof}
We consider only two rules:\\
Rule $\Fu\land$. Let us suppose that $\alpha\realizza \Fu(A\land B)$. By the meaning of
$\Fu$ we have that $\alpha\nonforza A\land B$ and for every world $\beta\in P$, if
$\alpha<\beta$, then $\beta\forza A\land B$. This implies that $\alpha\nonforza A$ or
$\alpha\nonforza B$ and $\beta\forza A$ and $\beta\forza B$. We have two main cases on $A$:
if $\alpha\nonforza A$ holds, then, since $\beta \forza A$, we get $\alpha\realizza \Fu
A$. Moreover, from $\beta\forza B$,  $\alpha\realizza\Tn B$ follows; if $\alpha\forza A$
holds, then $\alpha\realizza \T A$ holds and, by $\alpha \nonforza A\land B$, 
we have $\alpha\nonforza B$ and since
$\beta\forza B$ it  follows that $\alpha\realizza \Fu B$ holds;

\noindent
Rule $\Fn\Ttilde$. The correctness of the rule can be explained
following~\cite{AveFerMig:99}.
Let us suppose
that $\alpha\realizza S, \Ttilde(A_1\to B_1),\dots,\Ttilde(A_n\to
B_n),\Fn A_{n+1},\dots, \Fn A_u$. By the meaning of $\Ttilde$ we have
that $\alpha \realizza \Fn A_1,\T(A_1\to B_1),\dots, \Fn A_n,
\T(A_n\to B_n)$. 
Thus there exists $\beta_i$ such
that $\alpha<\beta_i$ and $\beta_i\realizza \Fu A_i$, for $i=1,\dots,u$. 
We notice that $\beta_i$ realizes all the $\T$ formulas
in $S$ and $\beta_i\realizza \T C$ if $\Fu C\in S$. 
Moreover, if $\beta_i=\min\{\beta_1,\dots,\beta_u\}$, 
then $\beta_i\realizza \F A_1,\dots,\F A_{i-1},\Fu A_i, \F A_{i+1},
\dots,\F A_u$. By the meaning of $\T, \That$ and $\Tn$ we conclude that  
if $i\in \{1,\dots,n\}$, then 
\[
\beta_i\realizza\{\That(A_1\to B_1),\dots,\That(A_n\to B_n), \F
A_{n+1},\dots \F A_u\}
\cup \{\Fu A_i,\Tn B_i\},
\] 
otherwise
%
%
$
\beta_i\realizza \{\That(A_1\to B_1),\dots,\That(A_n\to B_n), 
\F A_{n+1},\dots,\F A_u\}
\cup \{\Fu A_i \}.
$  
\end{proof}

From the proposition above we get
\begin{theorem}
  If there exists a closed proof table for $A$, then $A$ is valid
  in Dummett logic.
\end{theorem}

\section{Completeness}
We describe a procedure using the rules of the calculus to return a proof or a
counter model for a given set of signed formulas $S$.

In the following we sketch the recursive procedure \ftab(S). Given a
set $S$ of formulas, $\ftab(S)$ 
returns either a 
closed proof table for $S$ or $\NULL$ (if there exists a model realizing
$S$). 
To describe $\ftab$ we use the following definitions and notations. 
We call $\alpha$-rules and $\beta$-rules the rules of
Figure~\ref{fig:invertible-rules} with one conclusion and with two
conclusions, respectively. The $\alpha$-formulas and $\beta$-formulas are 
the kind of the signed formulas in evidence in the premise of the
$\alpha$-rules and $\beta$-rules, respectively (e.g. $\T(A\land B)$ is an
$\alpha$-formula and $\T(A\lor B)$ is a $\beta$-formula).  
Let $S$ be a set of
formulas, let $H\in S$ be an $\alpha$ or $\beta$-formula. 
With $Rule(H)$ we denote the rule corresponding to $H$ in
Figure~\ref{fig:invertible-rules}.
Let $S_1$ or 
$S_1 \EE S_2$ be the nodes of the proof tree obtained by applying to
$S$ the rule $Rule(H)$.  
If $Tab_1$ and $Tab_2$ are closed proof tables for $S_1$
and $S_2$ respectively, then $\NewTab{S}{Tab_1}{Rule(H)}$ or $\NewTab{S}{Tab_1 \EE
  Tab_2}{Rule(H)}$ denote the closed proof table for $S$ defined in the obvious
way. Moreover,  
$\Rcal_i(H)$ ($i=1,2$) denotes the set containing the formulas of
$S_i$ which replaces $H$. For
instance:\\ 
$\Rcal_1(\T (A\land B)) =\,\{\, \T A, \T B \,\}$,\\
$\Rcal_1(\T (A\lor B)) =\,\{ \T A\}$,  $\Rcal_2(\T (A\lor B)) =\,\{\T B \}$,\\
In the case of $\Fn\Ttilde$ we generalize the above notation.
Let $S_{\Fn\Ttilde}$ be the set of all the 
$\Fn$-formulas of $S$. Let $S_1|\dots|S_n$ be the nodes of the proof
tree obtained by applying to $S$ the rule $\Fn\Ttilde$. If
$Tab_1\dots,Tab_n$ are closed proof tables for
$S_1,\dots,S_n$, respectively, then $\NewTab{S}{Tab_1 \EE\dots\EE
  Tab_n}{\Fn\Ttilde}$ is the closed proof table for $S$. With $\Rcal_i(S_{\Fn\Ttilde})$ we denote the
set of formulas that replaces the set $S_{\Fn\Ttilde}$ in the $i$-th
conclusion of $\Fn\Ttilde$. For example, given
$S_{\Fn\Ttilde}=\{\Ttilde (A_1\to B_1),\Fn A_2,\Fn A_3\}$, 
$\Rcal_2(S_{\Fn\Ttilde})=\{\That(A_1\to B_1),\Fu A_2, \F A_3\}$.\\ 

\noindent
{\sc Function} \ftab(S)\\
\newcounter{passo}
\setcounter{passo}{1} 
\newcounter{counter:firstStep}
\setcounter{counter:firstStep}{\value{passo}}                  
{\bf   \arabic{passo}.} If $S$ is an inconsistent set, then $\ftab$ returns the proof $S$;
\\[1ex]
\addtocounter{passo}{1}
{\bf \arabic{passo}.} If an $\alpha$-rule applies to $S$, then let $H$ be a
$\alpha$-formula of $S$. If $\ftab((S\setminus\{H\})\cup \Rcal_1(H))$ returns a
proof $\pi$, then $\ftab$ returns the proof $\NewTab{S}{\pi}{Rule(H)}$,
otherwise $\ftab$ returns \NULL;\\[1ex]
\addtocounter{passo}{1}
{\bf \arabic{passo}.}  If a $\beta$-rule applies to $S$, then let $H$ be a
$\beta$-formula of $S$. 
Let $\pi_1=\ftab((S\setminus \{H \})\cup\Rcal_1(H))$
and $\pi_2=\ftab((S\setminus \{H\})\cup\Rcal_2(H))$. 
If $\pi_1$ or $\pi_2$ is \NULL, then 
$\ftab$ returns \NULL,  otherwise
$\ftab$
returns $\NewTab{S}{\pi_1\EE \pi_2}{Rule(H)}$;
\\[1ex]  
\addtocounter{passo}{1} 
\newcounter{counter:fn}
\setcounter{counter:fn}{\value{passo}} 
{\bf \arabic{passo}.}  If the rule $\Fn\Ttilde$ applies to $S$, then 
let $S_{\Fn\Ttilde}=\{
\calS A\in S | \calS\in\{\Ttilde,\Fn\}\}$ and $n=|S_{\Fn\Ttilde}|$.
If there exists $i\in\{1,\dots,n\}$, such that 
$\pi_i=\ftab((S\setminus S_{\Fn\Ttilde} )_c\cup\Rcal_i(S_{\Fn\Ttilde}))$ is \NULL, then 
$\ftab$ returns \NULL.  Otherwise $\pi_1,\dots,\pi_n$ are proofs 
and  $\ftab$
returns $\NewTab{S}{\pi_1\EE \dots\EE\pi_n}{\Fn\Ttilde}$;\\[1ex]
\addtocounter{passo}{1}                
\newcounter{counter:lastStep}
\setcounter{counter:lastStep}{\value{passo}}                  
{\bf \arabic{passo}.} If none of the previous points apply, then $\ftab$ returns
\NULL.\\[1ex]
{\sc end function $\ftab$}.\\

We emphasize that function $\ftab$ respects a particular sequence
in the application of the rules: $\Fn\Ttilde$ is applied if no
other rule is applicable. 
As a result no backtracking step is necessary.
Moreover, to decide $A$, the function call $\ftab(\{\Fu A \})$ is performed. By the rules
handling $\Fu$-formulas we have that when rule $\Fn\Ttilde$ is applied the formal
parameter $S$ contains at least a  $\Fu$-atomic formula and by rule $\Fn\Ttilde$,
every actual parameter of the recursive call performed in Step~\arabic{counter:fn},
contains a $\Fu$-formula and a  $\T$-atomic formula not occurring in $S$. 
This implies that every
application of $\Fn\Ttilde$ introduces a new $\T$-atomic formula and thus we can have at
most $n$ applications of rule $\Fn\Ttilde$, where  $n$ is the number of propositional
variables in $A$. Since between two applications of rule $\Fn\Ttilde$ we  
cannot have an infinite sequence of rule applications, we conclude that the function call
$\ftab(\{\Fu A\})$ always terminates. More formally, we can define a binary relation $\prec$ 
on sets of formulas defined as follows: 
$S'\prec S$ iff (i) the set of  $\T$-atomic formulas in $S'$ includes the set of atomic
formulas in $S$, or (ii) the set of  $\T$-atomic formulas in $S'$ coincides with the set of atomic
formulas in $S$ and the number of connectives in $S'$ is lower than in $S$ or $(iii)$ 
 the sets $S$ and $S'$ contain the same $\T$-atomic formulas, the same number of
 connectives  and they differ for the sign
 of a single formula $A$ such that $\T A\in S$ and $\Ttilde A\in S'$ or $\Ttilde A\in S$ and $\That
 A\in S'$ or $\F A\in S$ and $\Fu A\in S'$ or $\F A\in S$ and $\Fn A\in S'$.
 By inspecting the rules of the calculus it follows that  every recursive call is
 performed on a actual parameter $S'$ such that $S'\prec S$. By
 definition of $\prec$ and the fact that the sets only contain subformulas of the formula
 to be decided, every chain of recursive calls on non-inconsistent sets ends in a
 set only containing signed atomic formulas. This implies that function $\ftab$ terminates.

In order to get the completeness of $\ftab$, in the following it is proved
that given a set of 
formulas $S$, if the call of $\ftab(S)$ returns $\NULL$, then there is enough
information to build a model \mbox{$\uK=\kripke$} such that
$\rho\realizza S$.

\begin{theorem}[Completeness of \ftab]
  \label{theo:completeness}
  Let $A$ be a formula. If $A$ is valid in propositional Dummett logic, then $\ftab(\{\Fu
  A\})$ returns a proof.
\end{theorem}
\begin{proof}
To prove the theorem,  we consider a set $S$ of formulas and we prove that
if $\ftab(S)$ returns $\NULL$,  then there exists a Kripke model $\uK=\kripke$ such that 
$\rho\realizza S$. We get the statement of the theorem  by setting $S=\{\Fu
A\}$ and using the contrapositive.

We proceed by induction on the number of nested recursive calls. It is worth to note that
the construction of $\uK$ uses the sets of formulas involved in Step~4 or~5 of function
$\ftab$ as elements of $P$.  

\noindent
{\em Basis:} There are no recursive calls. Then Step~\arabic{counter:lastStep} has been
performed. We notice that $S$ 
is not inconsistent (otherwise Step~\arabic{counter:firstStep} would have been
performed). 
Indeed, 
$S$ only contains atomic formulas
signed with $\T$, $\Tn$, $\Fu$.
It is easy to prove that the model
$\uK=\kripke$, where $\rho=S$, $P=\{\rho\}$, $\rho\leq\rho$ and $\rho\forza p$
iff $\T p\in S$, realizes $S$.\\
                               
\noindent
{\em Step:}  By induction hypothesis we assume that the proposition holds for all sets
 $S'$ such that $\ftab(S')$ requires less than $n$ recursive calls. 
We prove the proposition holds for a set 
$S$ such that $\ftab(S)$ requires $n$ recursive calls by inspecting
all the possible cases where the procedure returns the 
$\NULL$ value. 

\noindent
{\em \NULL value returned performing Step~4}. 
By induction hypothesis there exists a model $\uK'=\langle P',\leq',\rho',\forza' \rangle$ such
that $\rho'\realizza (S\setminus
S_{\Fn\Ttilde})_c\cup \Rcal_j(S_{\Fn\Ttilde})$.
We have two cases:  if  the $j$-th formula in the enumeration of $S_{\Fn\Ttilde}$ is
$\Ttilde(A_j\to B_j)$, then 
$\Rcal_j(S_{\Fn\Ttilde})=$$(\{\That (A\to B) | \Ttilde(A\to B)\in
S_{\Fn\Ttilde}\}\setminus$$\{\That(A_j\to B_j)\})\cup \{\Fu A_j,\Tn B_j \}\cup\{\F C | \Fn C\in S_{\Fn\Ttilde} \}$. 
By $\rho'\realizza \Fu A_j,\Tn B$, we have $\rho\realizza \That(A_j\to B_j)$ and
$\rho\nonforza A_j$. We also have that for every $\That(A\to B)\in
\Rcal_j(S_{\Fn\Ttilde})$, $\rho'\nonforza' A$ and for every $\F C\in
\Rcal_j(S_{\Fn\Ttilde})$, $\rho'\nonforza C$. 
We build the following structure $\uK=\langle P,\leq,\rho,\forza \rangle$ such that 
\[
\begin{array}[t]{rcl}
P & = & P'\cup\{\rho\},\\ 
\leq & = & \leq'\ \cup\ \{(\rho,\alpha)|\alpha\in P'\},\\ 
\forza & = & \forza'\ \cup\ \{(\rho,p)| \T p\in S\},
\end{array}
\]
where we set $\rho=S$.
Since $\uK'$ is a Dummett model realizing $(S\setminus S_{\Fn\Ttilde})_c$, 
it follows that $\uK$ is a Dummett
model. As a matter of fact, $\rho'$ is the immediate successor of
$\rho$ and $\T A\in S$ implies $\T A\in (S\setminus S_{\Fn\Ttilde})_c$,
thus the forcing relation is preserved. 
This also implies that: 
$\rho\nonforza A_j$ holds, that, together with the facts $\rho'\forza' A_j\to B_j$ and
$\rho'\nonforza' A_j$, 
implies $\rho\realizza \Ttilde(A_j\to B_j)$; 
for every $\That(A\to B)\in
\Rcal_j(S_{\Fn\Ttilde})$, $\rho\nonforza A$ holds, that together with the facts
$\rho'\forza'A\to B$ and 
$\rho'\nonforza' A$, 
implies that $\rho\realizza \Ttilde(A\to B)$; 
 for every $\F C\in
\Rcal_j(S_{\Fn\Ttilde})$, $\rho\nonforza C$ holds, that together with the fact
$\rho'\nonforza C$ implies 
$\rho\realizza\Fn C$. 
Thus we have proved that $\rho\realizza S_{\Fn\Ttilde}$.
As regard the other formulas in $S$: if $\Fu A\in S$, then $A$ is an atomic formula and 
$\T A\in (S\setminus S_{\Fn\Ttilde})_c$. 
Since $S$ is not inconsistent (otherwise Step~1 would have been performed)  $\T
A\not\in S$ holds, this implies $\rho\nonforza A$; if 
$\Tn A\in S$ holds, then $\T A\in (S\setminus S_{\Fn\Ttilde})_c$ and hence
$\rho'\forza A$.  
Summarizing we conclude that $\rho\realizza S$.
\end{proof}
We remark that  following the construction of Theorem~\ref{theo:completeness}, 
it is straightforward how to modify function \ftab\ to get a function returning a proof or a
counter model. In particular, the proof puts in evidence  that a counter
model can be extracted by any branch of a tableau proof ending in a non-contradictory set
to which no further rule is applicable. By the fact that every application of $\Fn\Ttilde$
introduces in the conclusion a new propositional variable, it follows that if a formula $A$
is realizable, then $\ftab$ returns a counter model for $A$ having $n+1$ elements at most, where
$n$ is the number of propositional variables of $A$.
Finally, note that the elements of the counter model $\uK$ 
are sets of formulas only with the aim to simply the discussion in next section.

\section{Handling $\F$-formulas}
Now we start to discuss a calculus handling formulas signed with $\T$ and $\F$ only. We
present our ideas in two steps. First we introduce  calculus $\set$ 
having rules to handle the main connective of $\F$-formulas, this allow us to get rid of
$\That$-decide rule. Then we go a step further to
get our  final calculus $\et$. 

To handle $\F$-formulas by rules based on the main connective, 
it is necessary to introduce a machinery to determine,
given $\Ttilde(A\to B)$, if $A$ is forced.
Such a machinery is based on a notion similar to the boolean
satisfiability of a formula  in a model. 
Let $S$ be a non-inconsistent set of signed formulas and let $A$ be a formula,  we write
$S\models A$ iff $\T A\in S$, 
$\Ttilde A\in S$  
or one of the following conditions holds: 
(i)   $A=\true$;  
(ii) $A=B\land C$, $S\models B$ and $S\models C$; 
(iii)  $A=B\lor C$,  $S\models B$ or $S\models C$; 
(iv)   $A=B\to C$ and $\Fn A\not\in S$ and  if $S\models B$ then $S\models C$.
\\
We are interested to check if $S\models A$ holds when 
$\Fn\Ttilde$ is the only rule applicable to $S$.
The relation $\models$ aims to express via syntax the semantical notion of
realizability. In other words,
we are looking for a syntactical checking for  forcing and
non-forcing of a formula in a world of the
Kripke model built in the proof of Theorem~\ref{theo:completeness}. 
Relation $\models$ allows us to express such  a checking via the way the formulas are
handled in  the construction of the counter model. The construction has the properties 
suggesting that a new calculus managing  $\T$ and $\F$-formulas and a syntactical
checking based on $\models$ can be given. We start to show a relation between 
$\forza$ and $\models$ in the construction given in Theorem~\ref{theo:completeness}:
\begin{lemma}
  \label{lemma:persistenza}
  Let $S$ be a set occurring in the construction of the model
  $\uK=\kripke$ in the proof of Theorem~\ref{theo:completeness} and let 
  $\alpha\in P$. Let us suppose that $\alpha\realizza S$. Then: 
  \\
  (i)  if $\calS A\in S$,  with $\calS A\in \{\T,\Ttilde,\That \}$, then for every
  $\beta\in P$ such that $\alpha\leq\beta$, $\beta\models A$ and 
  \\
  (ii) if $\calS A \in S$, with $\calS A \in \{\F,\Fu,\Fn\}$, then $\alpha\not\models A$. 
\end{lemma}
\begin{proof}
  Note that (i) states that $\models$ is persistent.
  The proof proceeds by induction on $A$.
  \medskip
  \\
  Basis:  $A$ is an atomic formula. We have the cases $\T A$, $\Fu A$, $\Fn A$ and $\F A$.
  \\
  Case $\T A$. By construction of $\uK$, for every
  $\beta\in P$ such that $\alpha\leq\beta$, $\T A\in \beta$ holds and we get $\beta\models
  A$ by definition of $\models$.
  \\
  Case $\Fu A$. By construction $\Fu A\in \alpha$ and $\T A\not\in \alpha$. By definition
  of $\models$ we get $\alpha\not\models A$.
  \\
  Case $\Fn A$. By construction $\Fn A\in \alpha$. Since by construction there exists a
  subsequent set $S'$ of $\alpha$ such that $\Fu A\in S'$, it follows that $\T A\not\in
  S$, thus, by definition of $\models$ we get $\alpha\not\models A$.
  \\
  Case $\F A$. By construction $\Fu A\in \alpha$ or $\Fn A\in \alpha$ and we immediately
  get that $\alpha\not\models A$.
  \medskip
  \\
  Step: we proceed according to the outer connective of $A$.
  \\
  Case $\T A=\T(B\to C)$. By construction of proof in
  Theorem~\ref{theo:completeness} we have three cases: 
  (i) there is a subsequent set $S'$
  of $S$ such that $\T C\in S'$ and $\alpha\realizza S'$. For every
  $\beta\in P$ such that $\alpha\leq P$, $\beta\realizza \T C$ and by induction hypothesis
  we conclude $\beta\models C$; 
  (ii) there exists a subsequent set $S'$ such that $\Fu
  B,\Tn C\in S$. Thus $\alpha\realizza \Fu B,\Tn C$ and by induction
  hypothesis applied to $\Fu B$ we get $\alpha\not\models B$. Moreover by the construction
  we have that there exists a set $S''$ such that $\T
  C\in S''$ and for every $\beta\in P$ such that $\alpha<\beta$, $\beta\realizza \T C$. By
  induction hypothesis on $C$ we get that $\beta\models C$ that together
  $\alpha\not\models B$ proves that $\alpha\models B\to C$; 
  (iii) by construction there exist a
  subsequent  set $S'$ of $S$ and $\alpha$ such that $\Fu B,\Tn C\in S'$ and $\beta\in
  P$ such that $\alpha<\beta$ and $\beta\realizza S'$. By proceeding as in Point~(ii) we
  get that for every $\gamma\in P$ such that $\beta\leq \gamma$, $\gamma\models B\to C$. 
  Moreover for every $\gamma\in P$ such that $\alpha\leq\gamma$ and $\gamma<\beta$, 
  $\Ttilde A\in \gamma$. By definition of $\models$ we immediately get that 
  $\gamma\models \Ttilde(B\to C)$. Thus we have proved that for every $\beta\in
  P$ such that $\alpha\leq\beta$, $\beta\models B\to C$.

\end{proof}

Next Proposition~\ref{prop:forcingAndSatisfiability} is the main step to introduce
our new calculus. We express  the relationship between  
$\forza$ and $\models$ in the construction of the counter model
given in Theorem~\ref{theo:completeness}:
\begin{proposition}
  \label{prop:forcingAndSatisfiability}
  Let $S$ be a set and let us suppose that $\Fu A\in S$, the call $\ftab(S)$ returns \NULL
  and in the counter model $\uK=\langle P,\leq,\rho,\forza\rangle$ built in
  Theorem~\ref{theo:completeness}   there  exists an element of $P$ forcing $A$. Let $\alpha\in P$
  be the minimum world  
  such that $\alpha\forza A$. We have that
  $\alpha\models A$ and for every $\beta\in P$ such that $\beta<\alpha$, $\beta\not\models
  A$.
\end{proposition}
\begin{proof}
  By the construction given in Theorem~\ref{theo:completeness}, the hypothesis
  $\Fu A\in S$ implies $\rho\nonforza A$. 
  Moreover, by the meaning of the sign $\Fu$ we have that $\alpha$ is the immediate
  successor of $\rho$.
  What we are going to prove is that if in the
  construction of the counter model $\uK$ the formula $\Fu A$ is occurred and there exists 
  $\alpha\in P$ such that  $\alpha\forza A$, then the syntactical
  information in $\alpha$ allows us to prove $\alpha\forza A$ 
  via $\alpha\soddisfa A$ also when $\T A\not\in \alpha$.  
  We proceed by induction on $A$. \\
  Basis: $A$ is a propositional variable. We have that $\alpha\forza A$ iff
  (by definition of $\forza$) $\T A\in \alpha$ iff $\alpha\soddisfa A$ (by
  definition of $\soddisfa$). Moreover, for every $\beta<\alpha$, since
  $\beta\nonforza A$ we have that $\T A\not\in \beta$, thus
  $\beta\not\models A$.
  \\
  Step:
  \\
  Case $A=B\to C$. 
  In the stack of the recursive calls of $\ftab$(S) there exists a subsequent
  set $S'$ of $S$  such that $\T B\in S'$ and $\Fu C\in S'$. By the
  completeness theorem we have $\rho\realizza S'$, thus 
  $\rho\forza B$ and $\rho\nonforza C$. 
  Since $\alpha\forza B\to C$, from $\rho\forza B$  it follows
  that $\alpha\forza C$. 
  By induction hypothesis on $C$, $\alpha\models C$. Thus we conclude that
  $\alpha\models B\to C$.
  Since $\T B\in S'$, $\Fu C\in S'$ and $\rho\realizza S'$, by
  Lemma~\ref{lemma:persistenza} we get that 
  $\rho\not\models B\to C$ holds.
  \\
  
  \noindent
  Case $A=B\land C$. We have three cases: (i) in the stack of the recursive calls there
  exists a subsequent set $S'$ of $S$ such that $\Fu B,\Fu C\in S'$. By
  Theorem~\ref{theo:completeness} $\rho\realizza \Fu B,\Fu C$, thus $\rho\nonforza B$ and
  $\rho\nonforza C$. Since $\alpha\forza A\land B$ we have $\alpha\forza A$ and
  $\alpha\forza B$. 
  By induction hypothesis applied to $B$ and $C$ we get 
  $\alpha\models B$ and $\alpha\models C$, thus
  $\alpha\models B\land C$,  $\rho\not\models B$ and $\rho\not\models C$, thus
  $\rho\not\models B\land C$; 
  (ii) $\Fu B,\T C\in S'$. By Theorem~\ref{theo:completeness},
  $\rho\realizza \Fu B,\T C$. Since $\alpha\forza B\land C$, we have $\alpha\forza B$ and
  $\alpha\forza C$. By induction hypothesis and Lemma~\ref{lemma:persistenza}, $\alpha\models B\land C$. Moreover, since
  $\rho\realizza S'$ and $\Fu B\in S'$ we get $\rho\not\models B$. (iii) $\T B,\Fu C\in
  S$. The case goes as (ii).
  \\
  Case $A=B\lor C$. We have two cases: (i) $\F C,\Fu B\in S'$. By completeness theorem
  $\rho\realizza S'$. By hypothesis, $\alpha\realizza B$, thus by induction hypothesis
  applied to $B$, $\alpha\models B$ and $\rho\not\models B$. This
  implies  $\alpha \models B\lor
  C$. Now, since $\F C\in S'$ we have that there exists a set $S''$ subsequent to $S'$
  such that $\rho\realizza S''$ and $\Fu C\in S''$ or $\Fn C\in S''$. In both cases we
  get $\rho\not\models C$ and thus $\rho\not\models B\lor C$. (ii) $\F B, \Fu C\in
  S'$. The case goes as (i).
  
\end{proof}
Note that in the proof above, we take advantage from the fact that the world $\alpha$ is
the immediate successor of
$\rho$ and, as in Case $A=B\lor C$, we appeal to the fact $\rho\realizza S'$. The
difficult part will come when, by construction, we cannot say that  
the world $\alpha$ is the immediate successor of $\rho$. We will face this
problem with our final calculus $\et$. 
The strategy employed by function $\ftab$ implies that a $\F$-formula
sooner or later become a $\Fu$-formula. We can use the result above to get calculus \set,
which represents a first slight change to calculus \npwt: 
\begin{itemize}
\item We leave out the signs $\Ttilde$ and $\That$ and the rule
  $\That$-decide. 
\item the new calculus $\set$ has the rules $\T\land$, $\T\lor$,
  $\Fu\land$, $\Fu\lor$, $\Fu\to$ and $\F$-decide of $\npwt$.
  Rule $\Fn\Ttilde$ now becomes a rule handling $\Fn$-formulas only, thus we refer to it
  with the name of $\Fn$. Finally, $\set$ has the rules in Figure~\ref{fig:set};
  \begin{figure}[t]
    \[
    \begin{array}[t]{|rcl|}
      \hline
      \Tab{S,\T(A\to B)}{S,\T A\EE S,\F A,\To(A\to B)}{\T\to_1}
      &
      &
      \Tab{S,\To(A\to B)}{S,\T B}{\To\text{, provided $S\models A$}}\\
      \hline
    \end{array}
    \]
    \caption{Rules for \set}
    \label{fig:set}
  \end{figure}
\item relation $\models$ needs to be redefined according to the  syntax of the new
  calculus: Let $S$ be a set of signed formulas and $A$ a formula,  
  we write $S\models A$ iff $\T A\in S$,
  $\To A\in S$  
  or one of the following conditions hold: 
  (i)   $A=\true$;  
  (ii) $A=B\land C$ and $S\models B$ and $S\models C$; 
  (iii)  $A=B\lor C$ and  $S\models B$ or $S\models C$; 
  (iv)   $A=B\to C$, $\Fn A\not\in S$ and   if $S\models B$ then $S\models C$.

\item the sign  $\mathbf{\overline{\T}}$ is introduced to mark forced
  formulas of the kind $A\to B$ that are not at disposal of the rule $\T\to$ because
  already handled previously in the branch.  
  By the
  propositions given above, if $S\not\models A$ holds, then meaning of $\To(A\to B)$ is exactly
  the same of $\Ttilde(A\to B)$.  
\end{itemize}
By using previous results it is not difficult  design a decision
procedure based on \set and to prove correctness and completeness. In such a procedure
rule $\To$ is possibly applied if no other rule but $\Fn$ is applicable.

Now we can do another step and get rid of sign $\Fu$  and  rule $\F$-decide. 
The propositions given above use the fact that the information about an $\F$-formula is
not syntactically lost. As a matter of fact, 
every $\F$-formula is handled by $\F$-decide and sooner or
later a $\F$-formula is turned into a $\Fu$-formula and in the meantime the $\F$-formula
has become a $\Fn$-formula.

The rules of this new calculus $\et$ are given in
Figure~\ref{fig:thirdCalculus}. The calculus works on the signs $\T$
and $\F$. The sign $\To$ labels formulas that are not at disposal of deduction, thus it is
not part of the  object language. 
The signs $\Fu$, $\Tn$ and $\Fu$ are no longer necessary to
get a calculus obeying the subformula property. 
A set $S$ is inconsistent iff $\T\false\in S$ or $\{\T A,\F A \}\subseteq S$. 
Note rule $\F\land$ where both $A$ and $B$ 
occur. This is necessary to get for $\et$ the analogous of
Proposition~\ref{prop:forcingAndSatisfiability}. For this calculus
relation $\models$ is defined as follows:
$S\models A$ iff $\T A\in S$,
$\To A\in S$  
or one of the following conditions hold: 
(i) $A=B\land C$ and $S\models B$ and $S\models C$; 
(ii)  $A=B\lor C$ and  $S\models B$ or $S\models C$; 
(iii)   $A=B\to C$, $\F A\not\in S$ and   if $S\models B$ then $S\models C$.
As for the rules of the calculus, $\To$ is the only rule requiring a proof of
correctness. Moreover, for every rule of $\et$ but $\F\to$, it is immediate to check that
if an element $\alpha\in P$ of a model $\uK=\kripke$
realizes one of the sets in the
conclusion, then $\alpha$ also realizes the premise. 
\begin{figure}[h]
{
  \small
  \[
  \begin{array}[t]{|c|}
    \hline
    \begin{array}{ cccc }
      \Tab{S, \T (A \land B) }{S, \T A ,  \T B }{\T\land} 
      &
      \Tab{S,\F (A\land B)}{S,\F A,\F B|S,\F A,\T B | S, \T A, \F B}{\F\land}
      &
      \Tab{S, \T (A \lor B) }{S, \T A | S, \T B }{ \T\lor} 
      &
      \Tab{S,\F(A\lor B)}{S,\F A,\F B}{\F\lor} 
    \end{array}
    \\\\
    \begin{array}{cc}
      \Tab{S, \T (A \to B) }{S, \T B | S,\F A, \To (A\to B) }{\T\to_1} 
      &
      \Tab{S,\To(A\to B)}{S,\T B}{\To}\ \  
      \begin{minipage}[c]{0.5\linewidth}
        provided $S\models A$ 
      \end{minipage}
    \end{array}
    \\\\
    \Tab
    {
      S, \F (A_{1}\to B_1) ,\dots, \F (A_u\to B_u) 
    }
    { 
      S_c,V_{1} |\dots |S_c,V_u
    }
    { 
      \F\to
    }
    \\
    \text{where:}\hspace*{\fill}\\
    \begin{array}{l}
      \text{ for }j=1,\dots,u\\ 
      \begin{array}{ll}
        V_j = &(\{\F (A_{1}\to B_1),\dots, \F (A_u\to B_u)
        \}\setminus\{\F(A_j\to B_j) \})\cup \{\T A_j, \F B_j \}\\[1ex]
        S_c = &\{\T A \in S \} \cup \{\To A\in S\} 
      \end{array}
    \end{array}
    \\
    \hline
  \end{array}
  \]  
}
  \caption{The calculus \et}
  \label{fig:thirdCalculus}
\end{figure}
The following Function $\gtab$ uses calculus $\et$ to decide a set $S$. We recall that the
formulas in $S$ can be written only using $\T$ and $\F$, since $\To$ is a private
labelling of the deduction and as far as concerns the deduction $\To$-formulas are $\T$
formulas which are not at disposal of deduction.
\\ 
{\sc Function} \gtab(S)\\
\newcounter{bpasso}
\setcounter{bpasso}{1} 
\newcounter{bcounter:firstStep}
\setcounter{bcounter:firstStep}{\value{bpasso}}                  
{\bf   \arabic{bpasso}.} If $S$ is an inconsistent set, then $\gtab$ returns the proof $S$;
\\[1ex]
\addtocounter{bpasso}{1}
{\bf \arabic{bpasso}.} If an $\alpha$-rule applies to $S$, then let $H$ be a
$\alpha$-formula of $S$. If $\gtab((S\setminus\{H\})\cup \Rcal_1(H))$ returns a
proof $\pi$, then $\gtab$ returns the proof $\NewTab{S}{\pi}{Rule(H)}$,
otherwise $\gtab$ returns \NULL;\\[1ex]
\addtocounter{bpasso}{1}
{\bf \arabic{bpasso}.}  If a $\beta$-rule applies to $S$, then let $H$ be a
$\beta$-formula of $S$. 
Let $\pi_1=\gtab((S\setminus \{H \})\cup\Rcal_1(H))$
and $\pi_2=\gtab((S\setminus \{H\})\cup\Rcal_2(H))$. 
If $\pi_1$ or $\pi_2$ is \NULL, then 
$\gtab$ returns \NULL,  otherwise
$\gtab$
returns $\NewTab{S}{\pi_1\EE \pi_2}{Rule(H)}$;
\\[1ex]
\addtocounter{bpasso}{1}
{\bf \arabic{bpasso}.}  If  rule $\F\land$ applies to $S$, then let $H=\F(A\land B)$ be a
formula in $S$. 
Let $\pi_1=\gtab((S\setminus \{H \})\cup\{ \F A,\F B  \})$,
$\pi_2=\gtab((S\setminus \{H\})\cup\{\F A,\T B  \})$ and
$\pi_3=\gtab((S\setminus \{H\})\cup\{\T A,\F B  \})$. 
If $\pi_1$, $\pi_2$ or $\pi_3$ is \NULL, then 
$\gtab$\ returns \NULL,  otherwise
$\gtab$\
returns $\NewTab{S}{\pi_1\EE \pi_2 \EE \pi_3}{\F\land}$;
\\[1ex]  
\addtocounter{bpasso}{1} 
\newcounter{counter:To}
\setcounter{counter:To}{\value{bpasso}} 
{\bf \arabic{bpasso}.} If $\To(A\to B)\in S$ and $S\models A$, then
let $\pi_1=\gtab((S\setminus\{\To(A\to B)\})\cup\{\T B \})$.
If $\pi_1$ is \NULL then $\gtab$ returns \NULL, otherwise $\gtab$ returns $\NewTab{S}{\pi_1}{\To}$.
\\[1ex]
\addtocounter{bpasso}{1}
\newcounter{bcounter:f}
\setcounter{bcounter:f}{\value{bpasso}} 
{\bf \arabic{bpasso}.}  If the rule $\F\to$ applies to $S$, then 
let $S_{\F\to}=\{
\F( A\to B)\in S\}$ and $n=|S_{\F\to}|$.
If there exists $i\in\{1,\dots,n\}$, such that 
$\pi_i=\gtab(S_c\cup\Rcal_i(S_{\F\to}))$ is \NULL, then 
$\gtab$ returns \NULL.  Otherwise $\pi_1,\dots,\pi_n$ are proofs 
and  $\gtab$
returns $\NewTab{S}{\pi_1\EE \dots\EE\pi_n}{\F\to}$;\\[1ex]
\addtocounter{bpasso}{1}                
\newcounter{bcounter:lastStep}
\setcounter{bcounter:lastStep}{\value{bpasso}}                  
{\bf \arabic{bpasso}.} If none of the previous points apply, then $\gtab$ returns
\NULL.\\[1ex]
{\sc end function }.\\
We need to prove that  the properties of $\models$ still hold in the construction of $\gtab$.  
Following the lines of Lemma~\ref{lemma:persistenza}, we can prove that relation $\models$ is persistent:
\begin{lemma}
  \label{lemma:persistenza2}
  Let us suppose that $\T X\in S$. Then in the construction, for every subsequent set $S'$
  of $S$, we have that $S'\models X$.  
\end{lemma}
In the following lemma we sketch correctness and completeness of $\gtab$.
\begin{theorem}
  Let $S$ be a set of formulas. We have that:\\ 
  (i) if $\gtab(S)$ returns \NULL, then there exists a Kripke
  model $\uK=\kripke$ such $\rho\realizza S$;   
  \\
  (ii) if $\gtab(S)$ returns a proof, then for every Kripke model $\uK=\kripke$ and for
  every $\alpha\in P$, $\alpha\nonrealizza S$. 
\end{theorem}
\begin{proof}
  We proceed by induction on the number of nested recursive calls. Note that if function $\gtab$
  returns \NULL, the elements of the Kripke model we build are the sets
  of formulas involved in Steps~\arabic{bcounter:f} and~\arabic{bcounter:lastStep}.   
  \\
  {\em Basis:} There are no recursive calls.\\
  \noindent
  (i) If $G(S)$ returns \NULL, then  Step~\arabic{bcounter:lastStep} has been
  performed. We notice that $S$ 
  is not inconsistent (otherwise Step~\arabic{bcounter:firstStep} would have been
  performed). 
  Indeed, 
  $S$ only contains atomic formulas
  signed with $\T$ or $\F$.
  It is easy to prove that the model
  $\uK=\kripke$, where $\rho=S$, $P=\{\rho\}$, $\rho\leq\rho$ and $\rho\forza p$
  iff $\T p\in S$, realizes $S$.\\
  (ii) If $G(S)$ returns a proof, then Step~\arabic{bcounter:firstStep} is performed, thus
  $S$ is inconsistent and an inconsistent set is not realizable. 
  \\
  \noindent
  {\em Step:}  By induction hypothesis we assume that the proposition holds for all sets
  $S'$ such that $\gtab(S')$ requires less than $n$ recursive calls. 
  To prove the proposition holds for a set 
  $S$ such that $\gtab(S)$ requires $n$ recursive calls, one has to  inspect
  all the possible steps of $\gtab$. 
  \\
  {\em Let us suppose that Step~\arabic{counter:To} is performed.} Thus we have that $\To(X\to
  Y)\in S$ and $S\models X$. The call $\gtab((S\setminus\{\To(X\to Y)\})\cup \{\T Y \})$
  is performed. We have to analyze two main cases:\\
  (i) The call $\gtab((S\setminus\{\To(X\to Y)\})\cup \{\T Y \})$ returns \NULL. By
  induction hypothesis there is a model $\uK=\kripke$ such that $\rho\realizza S\setminus\{\To(X\to
  Y)\})\cup \{\T Y \}$, thus $\rho\realizza S$;
  \\
  (ii) The call $\gtab(S\setminus\{\To(X\to Y)\})\cup \{\T Y \})$ returns a proof. We have
  to show that the rule application is correct. 
  We want to prove that if $S\models X$, $\alpha$ is an element of a Kripke model such that
  $\alpha\forza p$ iff $\T 
  p\in S$ and $\alpha\realizza S$, then $\alpha\realizza\T X$.
  By construction we have that in the stack of the recursive calls there exists a set
  $S_0$ such that $\T(X\to Y)\in S_0$ and a subsequent set $S_1$ of $S_0$ such that $\F
  X\in S_1$. This means that $S_1\not\models X$.  
  \begin{claim}
    Let $U$ be a set of the construction, $Z$ a formula and $\beta$ an element of a Kripke
    model respectively meeting the conditions of $S$, $X$ and $\alpha$. We claim
    that:\\
    (i) if $U\not\models Z$ and $\beta\realizza U$, then $\beta\realizza \F Z$;
    \\
    (ii) if $U\models Z$ and $\beta\realizza U$, then $\beta\realizza \T Z$.
    \begin{proof}
      The proof of the claim goes by induction on $Z$:
      \\
      {\em Basis:}
      $Z$ is an atomic formula.\\
      (i) If $U\not\models Z$, then $\T Z\not\in U$ and by the relation of forcing defined
      on $\beta$ we have $\beta\realizza \F Z$;\\
      (ii) if $U\models Z$, then $\T Z\in U$, thus $\beta\realizza \T Z$.
      \\
      {\em Step:} we only prove the case $Z=K\to H$.\\
      (i) $U\not\models K\to H$. We have two cases: (a) $\F(K\to H)\in U$, thus we
      immediately get $\beta\realizza \F(K\to H)$; (b) $\F(K\to H)\not\in U$. Thus
      $U\models K$ and $U\not\models H$. By induction hypothesis $\beta\realizza \T K$ and
      $\beta\realizza \F H$ and we get $\alpha\realizza \F(K\to H)$;
      \\
      (ii) $U\models K\to H$. Thus $\F(K\to H)\not\in U$. Since in the stack of the
      recursive calls there exists a set $S_1$ such that $\F(K\to H)\in S_1$, then there
      exists a subsequent set $S_2$ of $S_1$ such that $\T K,\F H\in S_2$. Thus
      $S_2\models K$. By Lemma~\ref{lemma:persistenza2}, $U\models K$ and thus $U\models
      H$. By induction hypothesis $\beta\realizza \T H$ and thus $\beta\realizza \T(K\to H)$. 
      \end{proof}
    \end{claim}
    Now, since $\alpha\realizza\To(X\to Y)$ means $\alpha\forza X\to Y$, by the claim we get
  $\alpha\forza X$ and thus $\alpha\forza Y$, that is $\alpha\realizza\T Y$ (note that by
  construction $\alpha$ meets the conditions of the claim). 
  \\
  {\em Let us suppose that Step~\arabic{bcounter:f} is performed.} Note that in this case $S$
  contains atomic formulas, formulas of the kind $\To(A\to B)$, with $S\not\models A$, and
  $\F(A\to B)$. Point~(ii)  is an easy task, since it is based on the fact that rule
  $\F\to$ preserves the realizability (Point~(ii) corresponds to the proof correctness of
  rule $\F\to$).   
  As for Point~(i), by induction hypothesis there exists a Kripke model $\uK'=\langle
  P',\leq'$, $\rho',\forza'\rangle$ such that $\rho'$ realizes one of the set in the
  conclusion of the rule. 
  We build the following structure $\uK=\langle P,\leq,\rho,\forza \rangle$ such that 
  \[
  \begin{array}[t]{rcl}
    P & = & P'\cup\{\rho\},\\ 
    \leq & = & \leq'\ \cup\ \{(\rho,\alpha)|\alpha\in P'\},\\ 
    \forza & = & \forza'\ \cup\ \{(\rho,p)| \T p\in S\},
  \end{array}
  \]
  where we set $\rho=S$.
  The difficult part in proving  $\rho\realizza S$ is to show that if $\To(A\to B)\in S$,
  then $\rho\realizza \To(A\to
  B)$. Since if Step~\arabic{bcounter:f} is performed and $\To(A\to B)\in S$ then
  $S\not\models A$. Note that by construction, in the stack of recursive calls, there exists a
  previous set $S_0$ of $S$ such that $\F A\in S$. Now by proceeding as in the claim above
  we can prove that $\rho\nonforza A$ and this allow us to get that $\rho\forza A\to B$.\\
  An analogous argument has to be applied when Step~\arabic{bcounter:lastStep} is
  performed, since in this case $S$ can contain $\To$-formulas. 
  
\end{proof}
 By inspecting the rules of the calculus, it is easy to prove that the procedure
 terminates and the depth of the deductions is linear in  the size of the formula to be
 decided.
 
The check to decide if rule $\To$ has to be applied is performed 
on every $\To$-formula when no other rule but $\F\to$ or possibly $\To$ is applicable.
Thus before every application of $\F\to$ or $\T\to$ the check is performed. Note that
every application of $\F\to$ and $\T\to$ erases at least an implication, thus along a
branch the number of times that the check is performed is linear in the length of the proof. 
A single check requires a linear number of steps in the number of connectives in the
antecedent. Summarizing, along a branch to check if  $\models$ holds requires a quadratic
number of steps in the size of the formula to be proved.    
 
\section{Conclusions}
In this paper we have presented two tableau calculi for propositional Dummett logic
obeying to the subformula property and whose deductions have respectively quadratic and
linear  depth in the size of 
the formula to be decided. The papers presented in  literature lack of fulfilling all these features. 

Both calculi do not require backtracking and are based
on a multiple premise rule. The object language of calculus \npwt contains signs
to characterize the semantical status of ``forced/non-forced in the next possible world'' or
``this is last possible world where the formula is not known'', which are also employed
in~\cite{Fiorino:2011}. Calculus \et uses the signs $\T$ and $\F$, that  is the semantics of the
signed formulas is restricted to the  
forcing or non-forcing, and the proof is built-up without the necessity of any particular
labelling. Calculus \et has a straightforward translation into a sequent calculus.

Our completeness theorems prove that  calculi \npwt and \et allow to provide a
procedure returning a 
counter model or a proof. In particular, 
a feature of $\npwt$ is that from a failed proof of a formula $A$ it is possible to
extract a counter model for $A$ whose depth is  $n+1$ at most, with $n$ the number of
propositional variables occurring in $A$.
From a remark on the completeness of \npwt we get
calculus \et. Calculus \et shows that the semantics of Dummett
logic implies that deduction  conveys syntactical information about implicative formulas 
that can be used to drive the deduction by means of a 
fast computational check on some formulas which are possibly not at disposal of the deduction.

The multiple premise rules such as 
$\Fn\Ttilde$ and $\F\to$, which are analogous to
the multiple premise rule  introduced
in~\cite{AveFerMig:99}, have been criticized because they have an arbitrary number
of premises and thus they are supposed not to be suitable for  automated deduction. In 
papers~\cite{Fiorino:2010,Fiorino:2011} 
we showed that implementations of systems equipped with a rule analogous to $\Fn\Ttilde$
and $\F\to$ are far better than the  implementation based on
decomposition systems of~\cite{AvrKon:2001,Larchey-Wendling:2007}, 
which reduce the formulas to implicative atomic formulas and then
applies transitivity rules or procedures based on graph reachability.

We note that it is possible to add some rules to optimize the proof search. As an
example, by refining the completeness theorem for \et, follows that given 
$\T(A\to B)$, if $A$ does not contain implications, then we can turn $\T(A\to B)$ into
$\To(A\to B)$, thus saving an application of $\T\to$ still preserving the completeness. 
We believe that there are more general cases on the syntax on $A$ that allow to avoid an
useless application of rule $\T\to$. Moreover, since the sign of the occurrence of $A$ in
$\T(A\to B)$ is $\F$, it could be possible to apply our check to $\F$-formulas in order to
avoid also useless applications of $\F$-rules. 

As a future work, the first question is an investigation along the above line, that  could
be useful both to
deepen the understanding of the proof theory of Dummett logic and to design more efficient
decision procedures.  Another question is to
extend, if possible, the same technique to the first-order case of Dummett logic. 
Finally, currently we are investigating  
how to adapt these techniques employed for \et to propositional intuitionistic logic, whose Kripke
semantics is more complicated than Dummett logic. Our preliminary results
show that both the syntactical check and the strategy  are more involved than 
those given for \et.


\appendix
\section{My review of CLS reviewers}
 BERTRAND MEYER:
 \begin{quote}
  "Refereeing should be what it was before science publication turned
  into a business: scientists giving their polite but frank opinion on
  the work of other scientists." \href{http://cacm.acm.org/magazines/2011/11/138207-in-support-of-open-reviews-better-teaching-through-large-scale-data-mining/fulltext}{\small (CACM, Vol. 54 No. 11).}
\end{quote}

I submitted this paper to IJCAR 2012 and, in the present form to CSL 2012.  In both cases
it was rejected. Now it is my turn to give a review of reviewers and spend some words about
my experience as an author in proof-theory.
\\
I start with the facts: at CLS the paper had three reviewers. The first gave an accept and
was the only reviewer to read the paper. Reviewers 2 and 3 clearly read  the
introduction, at most, as anyone can understand from the general comments they give.
\begin{verbatim}
----------------------- REVIEW 2 ---------------------
PAPER: 49
TITLE: Terminating Calculi for Propositional Dummett Logic with Subformula Property
AUTHORS: Guido Fiorino

OVERALL RATING: -3 (strong reject)

This paper is presenting two new calculi for propositional Dummet logic aka
Goedel logic. This logic can be viewed both as an intermediate logic
(intuitionistic logic + axiom scheme (a->b)\/(b->a) ) or as a fuzzy logic
with operators over the unit interval.

While this is a nice paper in pure logic, it is not clear to me why
this paper is submitted to CSL.

1) The paper contains no motivation that relates to computer science
except for a reference to a famous 1991 paper [2] on simple consequence
relations. It is neither clear why [2] is called "recent" nor why it
is considered a CS motivation.

2) There already exist many calculi for this logic including
[1,3,8,9,10,11,14]

3) The paper contains no generic discussion why tableau calculi are
the right approach for Dummett logic. Given the simplicity of
the logic considered, and its simple semantic characterization in
terms of the unit interval (i.e., Goedel logic), one would expect that a
DPLL style procedure similar to standard SAT and SMT solving is more
efficient in practice. Reductions of fuzzy logics to arithmetic
solvers have been proposed by Haehnle and others in the 90ies.

4) There also is no methodological breakthrough which can be
generalized to other logics.

In conclusion I think the paper is lacking motivation.
----------------------- REVIEW 3 ---------------------
PAPER: 49
TITLE: Terminating Calculi for Propositional Dummett Logic with Subformula Property
AUTHORS: Guido Fiorino

OVERALL RATING: 1 (weak accept)

This paper describes two terminating calculi for propositional Goedel
Dummett logicwith subformula property which is not the important point as
subformula property can be always obtained by suitable choice of the
linguistic frame. The first calculus is completely straightforwardly
obtained from valuations in linearly ordered Kripke semantics, the claim on
the size of models is however trivial as only valuations of variables count
in Goedel-Dummett logics, they are projective. The second calculus is much
more interesting and the paper should concentrate on this. Furthermore the
paper has the deficiency for non-experts of providing no single example.
-----------------------------------------------------
\end{verbatim}
In my opinion the reviewers have a conflict interest and want to make space for their
papers, thus I consider them in bad faith.
\\
As regard review 3, he/she gives a borderline rating. Here we meet the first
characteristic of many reviews in proof-theory: ``the topic is not important''.  The
reviewer pretend of ignoring that there are many papers about calculi with the subformula
property and many authors consider this property important.  Statement ``suitable choice
of the linguistic frame'', means to have hypertableaux/hypersequents and/or labelled
systems. The advantage of my systems is in evidence in the introduction (see paragraph
starting with ``Papers [4] and [15] provide calculi ...'') but the reviewer has ignored my
considerations.
\\
This is one of the behaviours that I observed by reviewers in proof-theory: minimize the
idea and the interest of the problem, in order not to give importance to the whole paper,
even if there are many papers along the same line (note that at CLS 2003 a paper
addressing the same question was proposed and in all the quoted papers the efficiency or
the subformula property or the proof-system or the termination is addressed).
\\
Reviewer 2 is the typical coward that hide himself under anonymous review to make nasty
statements and to give a very bad mark without entering into technical details. The aim of
the reviewer is clear: to be sure that the paper is rejected, independently of the others
reviews. A strong reject implies that the paper contains technical errors that cannot be
clearly fixed. But here the review is not scientific and the program committee is responsible
for this (I wonder if the reviewer has read
the whole introduction or at least the abstract).
\\
The reviewer states that he/she does not understand my submission to the conference. To
understand the submission he/she should read CSL call for paper. The paper perfectly
matches the topic both in proof-theory and automated deduction. Point 2) is perfect to
understand the bad faith of the reviewer: the argument is that there are enough papers on
Dummett logic thus we do not need more. It's a pity, my paper is late! On this base, I
aspect that in
the future CSL will reject papers on Dummett/Goedel logic, {\bf independently of the name(s) of
  the author(s)}.  Also Point 3) deserves attention, because it is another typical
scheme to reject a/my paper: ``why to provide a calculus when there is a translation
into another logic?''  On this base we cannot have calculi for propositional
intuitionistic logic, since there exist translations in S4 or classical logic and so on
for many other logics. Variants of this are ``why do you use semantical techniques?'' and
``I don't like the presentation'' and, following the Point 4) ``the result is not
interesting because it cannot be generalized''.
\\

\noindent
I charge the
reviewers to have used anonymous review to
be unfair, biased and in bad faith instead of giving a frank scientific opinion. 
\\

\noindent
{\bf The problem is not the content of the paper but the name of the
author}. Proof-theory is a close world, a kind of private club made of some schools and
newcomers are not welcome. Thus can happen that also a trivial mistake as a typo is used as an
excuse to give the minimum rate and the original ideas are ignored. The result is that for
authors that are not part of the club it is almost impossible to have a paper accepted to a
conference, the timings to have a paper accepted on a journal are amplified and when the
papers is published it is not cited, also if pertinent.  
\\

\noindent
For these reasons I support the statement of Bertrand Meyer.
\end{document}